%% file: Degrading.tex
\documentclass[final,conference, twoside, 10pt]{IEEEtran}

\usepackage[hyperref]{Bamble}
\input{Degrading_Shortcuts}

\begin{document}

\title{Greedy-Merge Degrading has Optimal Power-Law}
\author{Assaf~Kartowsky and~Ido~Tal\\  
       Department of Electrical Engineering \\
Technion - Haifa 32000, Israel\\
E-mail: \{kartov@campus, idotal@ee\}.technion.ac.il
}
\maketitle

\begin{abstract}
Consider a channel with a given input distribution. Our aim is to degrade it to a channel with at most $L$ output letters. One such degradation method is the so called ``greedy-merge'' algorithm. We derive an upper bound on the reduction in mutual information between input and output. For fixed input alphabet size and variable $L$, the upper bound is within a constant factor of an algorithm-independent lower bound. Thus, we establish that greedy-merge is optimal in the power-law sense.
\end{abstract}

\section{Introduction}
\label{sec:Introduction}
In myriad digital processing contexts, quantization is used to map a large alphabet to a smaller one. For example, quantizers are an essential building block in receiver design, used to keep the complexity and resource consumption manageable. The quantizer used has a direct influence on the attainable code rate. 

Another recent application is related to polar codes~\cite{Arikan:PolarCodes}. Polar code construction is equivalent to evaluating the misdecoding probability of each channel in a set of synthetic channels. This evaluation cannot be carried out naively, since the output alphabet size of a synthetic channel is intractably large. One approach to circumvent this difficulty is to degrade the evaluated synthetic channel to a  channel with manageable output alphabet size \cite{Tal:Construct}\cite{Pedarsani:Construction}\cite{Tal:ConstructingPolarCodes}\cite{Gulcu:DMC}\cite{PeregTal:17p}\cite{Tal:Wiretap}.

Given a design parameter $L$, we degrade an initial channel to a new one with output alphabet size at most $L$. We assume that the input distribution is specified, and note that this degradation reduces the mutual information between the channel input and output. In both examples above, this reduction is roughly the loss in code rate due to quantization. We denote the smallest reduction possible by $\DIDegStar$. 

Let $|\calX|$ denote the channel input alphabet size, and treat it as a fixed quantity. We show that for any input distribution and any initial channel, $\DIDegStar=O(L^{-2/(|\calX|-1)})$. Moreover, this bound is attained efficiently, by the greedy-merge algorithm \cite{Tal:Construct}\cite{Gulcu:DMC}. This bound is tighter than the bounds derived in \cite{Pedarsani:Construction}, \cite{Tal:ConstructingPolarCodes}, \cite{Gulcu:DMC} and \cite{PeregTal:17p}. In fact, up to constant multipliers (dependent on $|\calX|$), this bound is the tightest possible. Namely, \cite{Tal:ModerateSizes} proves the existence of an input distribution and a sequence of channels for which $\DIDegStar=\Omega(L^{-2/(|\calX|-1)})$. Both bounds have $-2/(|\calX|-1)$ as the power of $L$, the same power-law. Note that for noisy channels and a relatively small $L$ our bound can be tightened~\cite{Nazer:Distilling}. See also \cite{Zhang:KLMeans}, which is especially relevant in the context of small $L$.

\section{Preliminaries}
We are given an input distribution and a discrete memoryless channel (DMC) $W: \calX \rightarrow \calY$. Both $|\calX|$ and $|\calY|$ are assumed finite. Let $X$ and $Y$ denote the random variables that correspond to the channel input and output, respectively. Denote the corresponding distributions $P_X$ and $P_Y$. Let $W(y|x)\triangleq \Prob{Y=y|X=x}$. For brevity, let $\pi (x)\triangleq \Prob{X=x}=P_X(x)$. Assuming further that $\calX$ and $\calY$ are disjoint, we abuse notation and denote $\Prob{X=x|Y=y}$ and $\Prob{Y=y}$ as $W(x|y)$ and $\pi(y)$, respectively. Without loss of generality, $\pi(x) > 0$ and $\pi(y) > 0$. We \emph{do not} assume that $W$ is symmetric. 

The mutual information between channel input and output is 
\[
I(W,P_X) \triangleq I(X;Y) = \sum_{x\in \calX} \eta(\pi (x))-\sum_{\substack{x \in \calX, \\ y \in \calY} } \pi(y) \eta(W(x|y)) \; ,
\]
where $\eta(p)\triangleq -p \log p$ for $p>0$, zero for $p=0$, and the logarithm is taken in the natural basis. We note that the input distribution does not necessarily have to be the one that achieves the channel capacity.

We now define the relation of degradedness between channels. A channel $Q: \calX \rightarrow \calZ$ is said to be (stochastically) \emph{degraded} with respect to a channel $W: \calX \rightarrow \calY$, and we write $Q \preccurlyeq W$, if there exists a channel $\Phi: \calY \rightarrow \calZ$ such that
\begin{equation}
\label{eq:Degrading}
Q(z|x)=\sum_{y \in \calY} W(y|x) \Phi(z|y) \; ,
\end{equation}
for all $x \in \calX$ and $z \in \calZ$. Note that as a result of the data processing theorem, $Q \preccurlyeq W$ implies $\DIDeg \triangleq I(W,P_X)-I(Q,P_X) \geq 0$.

Although mentioned before, let us properly define the \emph{optimal degrading loss} for a given pair $(W,P_X)$ as
\begin{equation}
\label{eq:DIDeg}
\DIDegStar \triangleq \min_{\substack{{Q,\Phi:Q \preccurlyeq W,} \\ {|Q| \leq L}} } I(W,P_X)-I(Q,P_X) \; ,
\end{equation}
where $|Q|$ denotes the output alphabet size of the channel $Q$. The optimizer $Q$ is the degraded channel that is ``closest" to $W$ in the sense of mutual information, yet has at most $L$ output letters.

\section{Main result}
\label{sec:UBDC}
Our main result is an upper bound on $\DIDegStar$, in terms of $|\calX|$ and $L$. This upper bound will follow from analyzing a sub-optimal\footnote{For the binary-input case, optimal degrading can be realized through dynamic programming \cite{Kurkoski:Quantization}\cite{Iwata:SMAWK}. For the non-binary case, we do not know of an efficient realization of optimal degrading.} degrading algorithm, called ``greedy-merge''.
In each iteration of greedy-merge,  we merge the two output letters $y_a,y_b \in \calY$ that result in the smallest decrease of mutual information between input and output, denoted $\DIDeg$. Namely, the intermediate channel $\Phi$ maps $y_a$ and $y_b$ to a new symbol, while all other symbols are unchanged by $\Phi$. This is repeated $|\calY|-L$ times, to yield an output alphabet size of $L$. By upper bounding the $\DIDeg$ of each iteration we obtain an upper bound on $\DIDegStar$. A key result is the following theorem, stating  that there exists a pair of output letters whose merger yields a ``small" $\DIDeg$. 

\begin{theorem}
\label{thm:DIDeg}
	Let a DMC $W:\calX \rightarrow \calY$ satisfy $|\calY|>2|\calX|$, and let the input distribution $P_X$ be fixed. There exists a pair $y_a,y_b \in \calY$ whose merger results in a channel $Q$ satisfying $\DIDeg = O\p{|\calY|^{-\frac{|\calX|+1}{|\calX|-1}}}$. In particular,
	\begin{equation}
	\label{eq:DIDegOrder}
	\DIDeg \leq   \mu(|\calX|)\cdot|\calY|^{-\frac{|\calX|+1}{|\calX|-1}} \; ,
	\end{equation}
	where,
	\[
	\mu(|\calX|) \triangleq \frac{\pi|\calX|}{\p{\sqrt{1+\frac{1}{2(|\calX|-1)}}-1}^2}    \p{\frac{2|\calX|}{\Gamma \p{1+\frac{|\calX|-1}{2}}}}^{\frac{2}{|\calX|-1}} \; ,
	\]
	and $\Gamma(\cdot)$ is the Gamma function.	 
\end{theorem}

Recall that \Cref{thm:DIDeg} is referring to the merger of a single pair of output letters. The following corollary is our main result, and is basically an iterative utilization of \Cref{thm:DIDeg}.
\begin{corollary}
	\label{cor:DIDegStar}
	Let a DMC $W:\calX \rightarrow \calY$ satisfy $|\calY|>2|\calX|$ and let $L \geq 2|\calX|$. Then, for any fixed input distribution $P_X$,
	\[
	\DIDegStar =\min_{\substack{{Q,\Phi:Q \preccurlyeq W,} \\ {|Q| \leq L}} } I(W,P_X)-I(Q,P_X)= O\p{L^{-\frac{2}{|\calX|-1}}} \; .
	\] 
	In particular, $\DIDegStar \leq \nu(|\calX|)\cdot L^{-\frac{2}{|\calX|-1}}$, where $\nu(|\calX|) \triangleq \frac{|\calX|-1}{2} \mu(|\calX|)$, and $\mu(\cdot)$ was defined in \Cref{thm:DIDeg}. This bound is attained by greedy-merge, and is tight in the power-law sense.
\end{corollary}
\begin{proof}
	If $L \geq |\calY|$, then obviously $\DIDegStar=0$ which is not the interesting case. If $2|\calX| \leq L <|\calY|$, then applying \Cref{thm:DIDeg} repeatedly $|\calY|-L$ times yields
	\begin{align*}
		\DIDegStar &\leq \sum_{\ell=L+1}^{|\calY|} \mu(|\calX|) \cdot \ell^{-\frac{|\calX|+1}{|\calX|-1}} \\
		&\leq \mu(|\calX|)\int_L^{|\calY|} \ell^{-\frac{|\calX|+1}{|\calX|-1}} \dd \ell \\
		&\leq \nu(|\calX|) \cdot  L^{-\frac{2}{|\calX|-1}} \; ,
	\end{align*}
	by the monotonicity of $\ell^{-(|\calX|+1)/(|\calX|-1)}$. The bound is tight in the power-law sense, by \cite[Theorem 2]{Tal:ModerateSizes}.
\end{proof}

Note that for large values of $|\calX|$, the Stirling approximation along with some other first order approximations can be applied to simplify $\nu(|\calX|)$ to $\nu(|\calX|) \approx  16\pi e |\calX|^3$.

\section{Proof of Theorem~\ref{thm:DIDeg}}
The proof of Theorem~\ref{thm:DIDeg} will follow from a sphere-packing argument. In the following subsections we define a ``distance'' function, overcome it not being a metric, and assign different ``weights'' to different spheres. See \cite{KartowskyTal:TightDegrading} for more commentary.
\subsection{An alternative ``distance" function} 
Consider the merger of a pair of output letters $y_a,y_b \in \calY$.  
The new output alphabet of $Q$ is $\calZ=\calY \setminus \ppp{y_a,y_b} \cup \ppp{y_{ab}}$. The channel $Q:\calX \rightarrow \calZ$ then satisfies $Q(y_{ab}|x)=W(y_a|x)+W(y_b|x)$, whereas for all $y \in \calZ \cap \calY$ we have $Q(y|x)=W(y|x)$. Using the shorthand
\[
\pi_{ab} = \pi(y_{ab}) \; , \quad \pi_a = \pi(y_a) \; , \quad \pi_b = \pi(y_b) \; ,
\]
one gets that $\pi_{ab}=\pi_a+\pi_b$. 
Denote by $\bfalpha=(\alpha_x)_{x \in \calX}$, $\bfbeta=(\beta_x)_{x \in \calX}$ and $\bfgamma=(\gamma_x)_{x \in \calX}$ the vectors corresponding to posterior probabilities associated with $y_a,y_b$ and $y_{ab}$, respectively. Namely, $\alpha_x = W(x|y_a)$, $\beta_x = W(x|y_b)$, and 
\begin{equation}
\label{eq:gamma}
\gamma_x = Q(x|y_{ab}) = \frac{\pi_a \alpha_x + \pi_b \beta_x }{\pi_{ab}}=\frac{\pi_a \alpha_x + \pi_b \beta_x }{\pi_a+\pi_b} \; . 
\end{equation}
Thus, after canceling terms, one gets that
\begin{equation}
\label{eq:DIxsum}
\DIDeg=I(W,P_X)-I(Q,P_X)=\sum_{x \in \calX} \Delta I_x \; ,
\end{equation}
where $\Delta I_x \triangleq \pi_{ab}\eta(\gamma_x)-\pi_a\eta(\alpha_x)-\pi_b\eta(\beta_x)$.

In order to bound $\DIDeg$, we give two bounds on $\Delta I_x$. The first bound was derived in~\cite{Gulcu:DMC},
\begin{equation}
\label{eq:d1}
	\Delta I_x \leq (\pi_a + \pi_b)\cdot d_1(\alpha_x,\beta_x) \; ,
\end{equation}
where for $\alpha \geq 0$ and $\zeta \in \reals$, we define $d_1 (\alpha,\zeta) \triangleq |\zeta-\alpha|$ .

The subscript ``$1$'' in $d_1$ is suggestive of the $L_1$ distance. 
We will use $\alpha$ to denote a probability associated with an input letter, while $\zeta$ will denote a ``free'' real variable, possibly negative. 
Note that the bound in \eqref{eq:d1} was derived assuming a uniform input distribution, however remains valid for the general case.

We now derive the second bound on $\Delta I_x$. For the case where $\alpha_x,\beta_x >0$,
\begin{align*}
	\Delta I_x &= \pi_a (\eta(\gamma_x)-\eta(\alpha_x)) + \pi_b(\eta(\gamma_x)-\eta(\beta_x)) \\
	&\stackrel{(a)}{\leq} \pi_a \eta'(\alpha_x)(\gamma_x-\alpha_x) + \pi_b \eta'(\beta_x)(\gamma_x-\beta_x) \\
	&\stackrel{(b)}{=} \frac{\pi_a\pi_b}{\pi_a+\pi_b}(\alpha_x-\beta_x)(\eta'(\beta_x)-\eta'(\alpha_x))\\
	&\stackrel{(c)}{\leq} \frac{1}{4}(\pi_a+\pi_b)(\alpha_x-\beta_x)^2 (-\eta''(\lambda)) \; ,
\end{align*}
where in $(a)$ we used the concavity of $\eta(\cdot)$, in $(b)$ the definition of $\gamma_x$ (see \eqref{eq:gamma}), and in $(c)$ the AM-GM inequality and the mean value theorem where $\lambda=\theta\alpha_x +(1-\theta)\beta_x$ for some $\theta \in [0,1]$. Using the monotonicity of $-\eta''(p)=1/p$ we get $-\eta''(\lambda) \leq 1 / \min(\alpha_x,\beta_x)$. Thus,
\begin{equation}
\label{eq:d2}
	\Delta I_x \leq (\pi_a+\pi_b) \cdot d_2(\alpha_x,\beta_x) \; ,
\end{equation}
where 
\[
d_2(\alpha,\zeta) \triangleq \begin{cases}
\frac{(\zeta-\alpha)^2}{\min(\alpha,\zeta)} & \alpha,\zeta > 0 \; , \\
\infty & \mbox{otherwise} \; .
\end{cases}
\]
The subscript ``$2$'' in $d_2$ is suggestive of the squaring in the numerator. Combining \eqref{eq:d1} and \eqref{eq:d2} yields
\begin{equation}
\label{eq:d}
	\Delta I_x \leq (\pi_a+\pi_b) \cdot d(\alpha_x,\beta_x) \; ,
\end{equation}
where
\begin{equation}
\label{eq:dScalarDef}
d(\alpha,\zeta) \triangleq \min(d_1(\alpha,\zeta),d_2(\alpha,\zeta)) \; .
\end{equation}
Returning to \eqref{eq:DIxsum} using \eqref{eq:d} we get
\begin{equation}
\label{eq:DIBound}
	\DIDeg \leq (\pi_a + \pi_b)|\calX| \cdot d(\bfalpha,\bfbeta)  \; ,
\end{equation}
where
\begin{equation}
\label{eq:dDef}
	d(\bfalpha,\bfzeta) \triangleq \max_{x \in \calX} d(\alpha_x,\zeta_x) \; .
\end{equation}
We note that we use $\max$ in (\ref{eq:dDef}) instead of a summation to simplify the upcoming derivations. Moreover, according to (\ref{eq:DIBound}), it suffices to show the existence of a pair that is ``close" in the sense of $d$, assuming that $\pi_a,\pi_b$ are also small enough.

Since we are interested in lowering the right hand side of \eqref{eq:DIBound}, we limit our search to a subset of $\calY$, as was done in~\cite{Gulcu:DMC}. Namely, $\calYSmall \triangleq \ppp{y \in \calY: \pi(y)\leq 2/|\calY|}$,
which implies
\begin{equation}
\label{eq:YSmallSize}
|\calYSmall| \geq \frac{|\calY|}{2} \; . 
\end{equation}
Hence, $\pi_a+\pi_b \leq 4/|\calY|$ and
\begin{equation}
\label{eq:DIBoundYSmall}
\DIDeg \leq \frac{4|\calX|}{|\calY|}\cdot d(\bfalpha,\bfbeta) \; .
\end{equation}

We still need to prove the existence of a pair $y_a,y_b \in \calYSmall$ that is ``close" in the sense of $d$. To that end, as in~\cite{Gulcu:DMC}, we would like to use a sphere-packing approach. A typical use of such an argument assumes a proper metric, yet $d$ is not a metric. Specifically, the triangle-inequality does not hold.
The absence of a triangle-inequality is a complication that we will overcome, but some care and effort are called for. Broadly speaking, as usually done in sphere-packing, we aim to show the existence of a critical ``sphere" radius, $\rcritical=\rcritical(|\calX|,|\calY|) > 0$. Such a critical radius will ensure the existence of $y_a,y_b\in\calYSmall$ with corresponding $\bfalpha$ and $\bfbeta$ for which $d(\bfalpha,\bfbeta) \leq \rcritical$.

\subsection{Non-intersecting ``spheres"}
We start by giving explicit equations for the ``spheres" corresponding to $d_1$ and $d_2$.
\begin{lemma}
	\label{lm:B1B2}
	For $\alpha \geq 0$ and $r>0$, define the sets $\calB_1, \calB_2$ as
\[
\calB_i(\alpha,r) \triangleq \{\zeta \in \reals : d_i(\alpha,\zeta) \leq r\} \; , \quad i \in \{1,2\} \; .
\]
Then,
\[
\calB_1(\alpha,r) = \{\zeta \in \reals : -r \leq \zeta-\alpha \leq r \}
\]
and
\begin{multline*}
\calB_2(\alpha,r) \\
= \{\zeta \in \reals :  - \sqrt{r^2/4 + \alpha \cdot r} + r/2 \leq \zeta-\alpha  \leq \sqrt{\alpha \cdot r} \} \; .
\end{multline*}
\end{lemma}
\begin{proof}
	Assume $\zeta \in \calB_1(\alpha,r)$. Then $\zeta$ satisfies $|\zeta-\alpha| \leq r$, which is equivalent to $-r \leq \zeta - \alpha \leq r$, and we get the desired result for $\calB_1(\alpha,r)$. Assume now $\zeta \in \calB_2(\alpha,r)$. If $\zeta \geq \alpha$, then $\min(\alpha,\zeta)=\alpha$, and thus $(\zeta-\alpha)^2/\alpha \leq r$, which implies $0 \leq \zeta - \alpha \leq \sqrt{\alpha\cdot r}$. If $\zeta \leq \alpha$, then $\min(\alpha,\zeta)=\zeta$, and thus, $(\zeta-\alpha)^2/\zeta \leq r$, which implies $- \sqrt{r^2/4 + \alpha \cdot r} + r/2 \leq \zeta-\alpha  \leq 0$. The union of the two yields the desired result for $\calB_2(\alpha,r)$.
\end{proof}

Thus, we define $\calB(\alpha,r) \triangleq \{\zeta \in \reals : d(\alpha,\zeta) \leq r\}$, and note that $\calB(\alpha,r)=\calB_1(\alpha,r) \cup \calB_2(\alpha,r)$,
since $d$ takes the $\min$ of the two distances. Namely,
\begin{equation}
\label{eq:calBscalar}
\calB(\alpha,r) = \left\{\zeta \in \reals :  -\omegadown(\alpha,r) \leq \zeta-\alpha \leq \omegaup(\alpha,r) \right\} \; ,
\end{equation}
where $\omegadown(\alpha,r) \triangleq \max\left(\sqrt{r^2/4 + \alpha \cdot r} - r/2,r\right)$ and $\omegaup(\alpha,r) \triangleq  \max\left(\sqrt{\alpha \cdot r},r\right)$. To extend $\calB$ to vectors, we define $\realsX$ as the set of vectors with real entries that are indexed by $\calX$, $\realsX \triangleq \ppp{\bfzeta=(\zeta_x)_{x\in\calX}:\zeta_x \in \reals}$.
The set $\realsKX$ is defined as the set of vectors from $\realsX$ with entries summing to $1$, $\realsKX \triangleq \ppp{\bfzeta \in \realsX: \sum_{x\in\calX} \zeta_x=1}$ .
The set $\realsKplusX$ is the set of probability vectors. Namely, the set of vectors from $\realsKX$ with non-negative entries, $\realsKplusX \triangleq \ppp{\bfzeta \in \realsKX: \zeta_x \geq 0}$.
We can now define $\calB(\bfalpha,r)$. For $\bfalpha \in \realsKplusX$ let
\begin{equation}
\label{eq:calBVectorDef}
\calB(\bfalpha,r) \triangleq \ppp{\bfzeta \in \realsX:d(\bfalpha,\bfzeta) \leq r } \; .
\end{equation}
Using \eqref{eq:dDef} and \eqref{eq:calBscalar} we have a simple characterization of $\calB(\bfalpha,r)$ as a box: a Cartesian product of segments. That is, 
\begin{align}
\label{eq:calBVector}
\begin{split}
\calB(\bfalpha,r) &= \Big\{ \bfzeta \in \realsX :  \\
& \quad \quad -\omegadown(\alpha_x,r) \leq \zeta_x-\alpha_x \leq \omegaup(\alpha_x,r) \Big\} \; .
\end{split}
\end{align}
We stress that the box $\calB(\bfalpha,r)$ contains $\bfalpha$, but is not necessarily centered at it.

Recall our aim is finding an $\rcritical$. Using our current notation, $\rcritical$ must imply the existence of distinct $y_a, y_b \in \calYSmall$ such that $\bfbeta \in \calB(\bfalpha,\rcritical)$.
Note that the set $\calB(\bfalpha,r)$ is contained in $\realsX$. However, since the boxes are induced by points $\bfalpha$ in the subspace $\realsKplusX$ of $\realsX$, the sphere-packing would yield a tighter result if performed in $\realsKX$ rather than in $\realsX$. Then, for $\bfalpha \in \realsKplusX$ and $r>0$, let us define
\begin{equation}
\label{eq:calBK}
\calBK(\bfalpha,r) = \calB(\bfalpha,r) \cap \realsKX \; .
\end{equation}
When considering $\calBK(\bfalpha,r)$ in place of $\calB(\bfalpha,r)$, we have gained in that the affine dimension (see \cite[Section 2.1.3]{Boyd:Convex})
of $\calBK(\bfalpha,r)$ is $|\calX|-1$ while that of $\calB(\bfalpha,r)$ is $|\calX|$. However, we have lost in simplicity: the set $\calBK(\bfalpha,r)$ is not a box. Indeed, a moment's thought reveals that any subset of $\realsKX$ with more than one element cannot be a box.

We now show how to overcome the above loss. That is, we show a subset of $\calBK(\bfalpha,r)$ which is --- up to a simple transform --- a box. Denote the index of the largest entry of a vector $\bfalpha \in \realsKX$ as $\xmax(\bfalpha)$, namely, $\xmax(\bfalpha) \triangleq \argmax_{x \in \calX} \alpha_x$. In case of ties, define $\xmax(\bfalpha)$ in an arbitrary yet consistent manner. For $\xmax = \xmax(\bfalpha)$ given, or clear from the context, define $\bfzeta'$ as $\bfzeta$, with index $\xmax$ deleted. That is, for a given $\bfzeta \in \realsKX$, $\bfzeta' = (\zeta_x)_{x \in \calX'} \in \realsXminusone$, where $ \calX' \triangleq \calX \setminus \{\xmax\}$.
Note that for $\bfzeta \in \realsKX$, all the entries sum to one. Thus, given $\bfzeta'$ and $\xmax$, we know $\bfzeta$. Next, for $\bfalpha \in \realsKplusX$ and $r>0$, define the set
\begin{multline}
\label{eq:calC}
\calC(\bfalpha,r) = \{ \bfzeta \in \realsKX  :\\
 \forall x \in \calX' \; , \;  -\omega'(\alpha_x,r) \leq \zeta_x - \alpha_x \leq \omega'(\alpha_x,r) \} \; , 
\end{multline}
where $\xmax = \xmax(\bfalpha)$ and
\begin{equation}
\label{eq:omegaPrime}
\omega'(\alpha,r) \triangleq \frac{\omegadown(\alpha,r)}{|\calX|-1} \; .
\end{equation}

\begin{lemma}
\label{lemm:CinB}
	Let $\bfalpha \in \realsKplusX$ and $r > 0$ be given. Let $\xmax = \xmax(\bfalpha)$. Then, $\calC(\bfalpha,r) \subset \calBK(\bfalpha,r)$.
\end{lemma}
\begin{proof}
	It can be easily shown that $0 \leq \omegadown(\alpha,r) \leq \omegaup(\alpha,r)$. Thus, since \eqref{eq:calC} holds, it suffices to show that
	\begin{equation}
	\label{eq:omegaConditionsOnXmax}
	-\omegadown(\alpha_{\xmax},r) \leq \zeta_{\xmax} - \alpha_{\xmax} \leq \omegadown(\alpha_{\xmax},r) \; .
	\end{equation}
	Indeed, summing the condition in \eqref{eq:calC} over all $x \in \calX'$ gives
	\[
	\sum_{x \in \calX'} -\omega'(\alpha_x,r) \leq \sum_{x \in \calX'}\zeta_x -\sum_{x \in \calX'} \alpha_x \leq \sum_{x \in \calX'} \omega'(\alpha_x,r) \; .
	\]
	Since $\omegadown(\alpha,r)$ is a monotonically non-decreasing function of $\alpha$, we can simplify the above to
	\[
	-\omegadown(\alpha_{\xmax},r) \leq \sum_{x \in \calX'}\zeta_x -\sum_{x \in \calX'} \alpha_x \leq \omegadown(\alpha_{\xmax},r) \; .
	\]
	Since both $\bfzeta$ and $\bfalpha$ are in $\realsKX$, the middle term in the above is $\alpha_{\xmax}- \zeta_{\xmax}$. Thus, \eqref{eq:omegaConditionsOnXmax} follows. 
\end{proof}

Recall that our plan is to ensure the existence of a ``close" pair by using a sphere-packing approach. However, since the triangle inequality does not hold for $d$, we must use a somewhat different approach. Towards that end, define the positive quadrant associated with $\bfalpha$ and $r$ as
\begin{multline*}
\calQ'(\bfalpha,r) = \{ \bfzeta' \in \realsXminusone  : \\
\forall x \in \calX',\;  0 \leq \zeta_x - \alpha_x \leq \omega'(\alpha_x,r)  \} \; , 
\end{multline*}
where $\xmax = \xmax(\bfalpha)$ and $\omega'(\alpha,r)$ is as defined in \eqref{eq:omegaPrime}.

\begin{lemma}
	\label{lm:QSpheresNonintersecting}
	Let $y_a, y_b \in \calY$ be such that $\xmax(\bfalpha) = \xmax(\bfbeta)$. If $\calQ'(\bfalpha,r)$ and $\calQ'(\bfbeta,r)$ have a non-empty intersection, then $d(\bfalpha,\bfbeta) \leq r$.
\end{lemma}
\begin{proof}
	By (\ref{eq:calBVectorDef}), (\ref{eq:calBK}), and Lemma~\ref{lemm:CinB}, it suffices to prove that $\bfbeta \in \calC(\bfalpha,r)$. Define $\calC'(\bfalpha,r)$ as the result of applying a prime operation on each member of $\calC(\bfalpha,r)$, where $\xmax = \xmax(\bfalpha)$. Hence, we must equivalently prove that $\bfbeta' \in \calC'(\bfalpha,r)$. By \eqref{eq:calC}, we must show that for all $x \in \calX'$,
	\begin{equation}
	\label{eq:alphaxbetaxdistance}
	-\omega'(\alpha_x,r) \leq \beta_x - \alpha_x \leq \omega'(\alpha_x,r) \; .
	\end{equation}
	
	Since we know that the intersection of $\calQ'(\bfalpha,r)$ and $\calQ'(\bfbeta,r)$ is non-empty, let $\bfzeta'$ be a member of both sets. Thus, we know that for $x\in \calX'$, $0 \leq \zeta_x - \alpha_x \leq \omega'(\alpha_x,r)$, and $0 \leq \zeta_x - \beta_x \leq \omega'(\beta_x,r)$.
	For each $x\in \calX'$ we must consider two cases: $\alpha_x \leq \beta_x$ and $\alpha_x > \beta_x$.
	
	Consider first the case $\alpha_x \leq \beta_x$. Since $\zeta_x - \alpha_x \leq \omega'(\alpha_x,r)$ and $\beta_x - \zeta_x \leq 0$, we conclude that $\beta_x - \alpha_x \leq \omega'(\alpha_x,r)$. Conversely, since $\beta_x - \alpha_x \geq 0$ and, by \eqref{eq:omegaPrime}, $\omega'(\alpha_x,r) \geq 0$, we have that $\beta_x - \alpha_x \geq - \omega'(\alpha_x,r)$. Thus we have shown that both inequalities in \eqref{eq:alphaxbetaxdistance} hold.
	
	To finish the proof, consider the case $\alpha_x > \beta_x$. We have already established that $\omega'(\alpha_x,r) \geq 0$. Thus, since by assumption $\beta_x - \alpha_x \leq 0$, we have that $\beta_x - \alpha_x \leq \omega'(\alpha_x,r)$. Conversely, since $\zeta_x - \beta_x \leq \omega'(\beta_x,r)$ and $\alpha_x - \zeta_x \leq 0$, we have that $\alpha_x - \beta_x \leq \omega'(\beta_x,r)$. We now recall that by \eqref{eq:omegaPrime}, the fact that $\alpha_x \geq \beta_x$ implies that $\omega'(\beta_x,r) \leq \omega'(\alpha_x,r)$. Thus, $\alpha_x - \beta_x \leq \omega'(\alpha_x,r)$. Negating gives $\beta_x - \alpha_x \geq -\omega'(\alpha_x,r)$, and we have once again proved the two inequalities in \eqref{eq:alphaxbetaxdistance}.
\end{proof}

\subsection{Weighted ``sphere"-packing}
The volume of our ``sphere" $\calQ'(\bfalpha,r)$ unfortunately depends on $\bfalpha$. We would like then to alleviate this dependency by defining a density over $\realsXminusone$ and derive a lower bound on the weight of $\calQ'(\bfalpha,r)$. Let $\density : \reals \to \reals$ be defined as $\density(\zeta) \triangleq 1/\sqrt{4\zeta}$.
Next, for $\bfzeta' \in \realsXminusone$, abuse notation and define $\density : \realsXminusone \to \reals$ as $\density(\bfzeta') \triangleq \prod_{x \in \calX'} \density(\zeta_x)$.
The weight of $\calQ'(\bfalpha,r)$ is then defined as $\weightQTag \triangleq \int_{\calQ'(\bfalpha,r)} \density \dd \bfzeta'$.
The following lemma proposes a lower bound on $\weightQTag$ that does not depend on $\bfalpha$.
\begin{lemma}
	\label{lm:weightQTagLowerBound}
	The weight $\weightQTag$ satisfies
	\begin{equation}
	\label{eq:weightQTagLowerBound}
	\weightQTag \geq r^{\frac{|\calX|-1}{2}}\p{\sqrt{2+\frac{1}{|\calX|-1}}-\sqrt{2}}^{|\calX|-1}  .
	\end{equation}
\end{lemma}
\begin{proof}
	Since $\density(\bfzeta')$ is a product,
	\begin{align*}
	\weightQTag &= \prod_{x\in\calX'} \int_{\alpha_x}^{\alpha_x+\omega'(\alpha_x,r)} \frac{\mathrm{d}\zeta_x}{2\sqrt{\zeta_x}} = \prod_{x\in\calX'} \psi_r(\alpha_x) \; ,
	\end{align*}
	where $\psi_r(\alpha) \triangleq \sqrt{\alpha+\omega'(\alpha,r)}-\sqrt{\alpha}$. It can be shown that $\psi_r(\alpha)$ is decreasing when $\alpha<2r$ simply by using the first derivative. As for $\alpha \geq 2r$, it can be shown that $\psi_r'(\alpha)$ is non-zero. Since $\psi_r'(2r)>0$ we conclude that $\psi_r(\alpha)$ is increasing. By continuity we conclude that $\psi_r(\alpha)$ is minimal for $\alpha=2r$ and thus we get (\ref{eq:weightQTagLowerBound}).
\end{proof}

We divide the letters in $\calYSmall$ to $|\calX|$ subsets, according to their $\xmax$ value. The largest subset is denoted by $\calY'$, and we henceforth fix $\xmax$ accordingly. We limit our search to $\calY'$. 

Let $\calV'$ be the union of all the quadrants corresponding to possible choices of $\bfalpha$. Namely,
\[
\calV' \triangleq \bigcup_{\substack{\bfalpha \in \realsKplusX \; , \\ x_{\max}(\bfalpha)=x_{\max} } } \calQ'(\bfalpha,r)  \; .
\]
In order to bound the weight of $\calV'$, we introduce the simpler set $\calU'$.
\[
\calU' \triangleq \ppp{ \bfzeta' \in \realsXminusone: \sum_{x\in\calX'} \zeta_x \leq 2, \; \zeta_x \geq 0 \; \forall x\in\calX' } \;.
\]
The constraint $r \leq 1$ in the following lemma will be motivated shortly.
\begin{lemma}
	\label{lm:VSubsetU}
	Let $r \leq 1 $. Then, $\calV' \subseteq \calU'$.
\end{lemma} 
\begin{proof}
	Assume $\bfzeta'\in\calV'$. Then, there exists $\bfalpha \in \realsKplusX$ such that $0 \leq \zeta_x -\alpha_x \leq \omega'(\alpha_x,r)$ for all $x \in \calX'$. Hence, $\zeta_x \geq 0$  for all $x \in \calX'$. Moreover,
	\begin{align}
	\label{eq:sumOfZetaTag}
	\sum_{x \in \calX'} \zeta_x &\leq \sum_{x \in \calX'} \alpha_x + \sum_{x \in \calX'} \omega'(\alpha_x,r) \nonumber \\
	&\leq 1-\alpha_{x_{\max}} + \omegadown(\alpha_{\xmax},r) \;.
	\end{align}
	There are two cases to consider. In the case where $\alpha_{x_{\max}} \geq 2r$ we have
	\begin{align*}
	\sum_{x \in \calX'} \zeta_x &\leq 1-\alpha_{x_{\max}} + \sqrt{\frac{r^2}{4} + \alpha_{x_{\max}}r} - \frac{r}{2} \\
	&\leq 1-\alpha_{x_{\max}} + \sqrt{\frac{\alpha_{x_{\max}}^2}{16} + \frac{\alpha_{x_{\max}}^2}{2}} - \frac{r}{2} \\ 
	& \leq 2 \;,
	\end{align*}
	where the second inequality is due to the assumption $\alpha_{x_{\max}} \geq 2r$. In the case where $\alpha_{x_{\max}} \leq 2r$, \eqref{eq:sumOfZetaTag} becomes
	\[
	\sum_{x \in \calX'} \zeta_x \leq 1-\alpha_{x_{\max}} + r \leq 2-\alpha_{x_{\max}} \leq 2 \; ,
	\]
	where we assumed $r\leq 1$. Therefore, $\bfzeta' \in \calU'$.
\end{proof}

The lemma above and the non-negativity of $\varphi$, enable us to upper bound the weight of $\calV'$, denoted by $\weightVTag $, using $\weightVTag \triangleq \int_{\calV'} \varphi \dd \bfzeta' \leq \int_{\calU'} \varphi \dd \bfzeta'$.
We define the mapping $\rho_x = \sqrt{\zeta_x}$ for all $x \in \calX'$ and perform a change of variables. As a result, $\calU'$ is mapped to $\calS' \triangleq \ppp{ \bfrho'  \in \realsXminusone: \sum_{x \in \calX'} \rho_x^2 \leq 2, \; \rho_x \geq 0 }$,
which is a quadrant of a $|\calX|-1$ dimensional ball of a $\sqrt{2}$ radius. The density function $\varphi$ transforms into the unit uniform density function since $\dd \zeta_x/\sqrt{4\zeta_x} = \dd \rho_x$. Hence, for $r \leq 1$,
\begin{equation}
\label{eq:VWeight}
\weightVTag \leq \int_{\calS'} \dd V =\left(\frac{\pi}{2}\right)^{\frac{|\calX|-1}{2}} \frac{1}{\Gamma\p{1+\frac{|\calX|-1}{2}}} \; ,
\end{equation}
where we have used the well known expression for the volume of a multidimensional ball. Finally, we prove \Cref{thm:DIDeg}.
\begin{proof}[Proof of \Cref{thm:DIDeg}]
	Recall that we are assuming $|\calY| >2|\calX|$. According to the definition of $\calY'$, we get by \eqref{eq:YSmallSize} that
	\begin{equation}
	\label{eq:calYTagSize}
	|\calY'| \geq \frac{|\calYSmall|}{|\calX|} \geq \frac{|\calY|}{2|\calX|} > 1 \; .
	\end{equation}
As a result, we have at least two points in $\calY'$, and are therefore in a position to apply a sphere-packing argument. Towards this end, let $r$ be such that the starred equality in the following derivation holds:
	\begin{align}
\begin{split}
\label{eq:spherePackingDerivation}
		\sum_{\bfalpha \in \calY'} &\weightQTag  \\
		&\geq \frac{|\calY|}{2|\calX|}\cdot r^{\frac{|\calX|-1}{2}}\p{\sqrt{2+\frac{1}{|\calX|-1}}-\sqrt{2}}^{|\calX|-1} \\
		&\stackrel{(\ast)}{=} \p{\frac{\pi}{2}} ^{\frac{|\calX|-1}{2}} \frac{1}{\Gamma\p{1+\frac{|\calX|-1}{2}}} \\
		&\geq \weightVTag \; .
\end{split}
	\end{align}
Namely,
	\begin{multline}
		\label{eq:rstarDef}
	r \triangleq \frac{\pi}{4}\p{\sqrt{1+\frac{1}{2(|\calX|-1)}}-1}^{-2} \\
	\cdot \p{\frac{2|\calX|}{\Gamma\p{1+\frac{|\calX|-1}{2}}}}^{\frac{2}{|\calX|-1}}  \cdot  |\calY|^{-\frac{2}{|\calX|-1}} \; .
	\end{multline}
There are two cases to consider. If $r \leq 1$, then all of (\ref{eq:spherePackingDerivation}) holds, by \eqref{eq:weightQTagLowerBound}, \eqref{eq:VWeight} and \eqref{eq:calYTagSize}. We take $\rcritical = r$, and deduce the existence of a pair $y_a,y_b \in \calY'$ for which $d(\bfalpha,\bfbeta) \leq r$. Indeed, assuming otherwise would contradict (\ref{eq:spherePackingDerivation}), since each $\calQ'$ in the sum is contained in $\calV'$, and, by \Cref{lm:QSpheresNonintersecting} and our assumption, all summed $\calQ'$ are disjoint.
 
We next consider the case $r > 1$. Now, any pair of letters $y_a,y_b \in \calY'$ satisfies $d(\bfalpha,\bfbeta) \leq r$. Indeed, by \eqref{eq:dScalarDef} and \eqref{eq:dDef},
\[
d(\bfalpha,\bfbeta) \leq \|\bfalpha-\bfbeta\|_{\infty} \leq 1 < r \; ,
\]
where $\|\cdot\|_{\infty}$ is the maximum norm.

We have proved the existence of $y_a,y_b \in \calY' \subset \calYSmall$ for which $d(\bfalpha,\bfbeta) \leq r$. By  (\ref{eq:DIBoundYSmall}) and (\ref{eq:rstarDef}), the proof is finished.
\end{proof}

\end{document}

%% file: Degrading_Shortcuts.tex
\newcommand {\calX} {\mathcal{X}}
\newcommand {\calY} {\mathcal{Y}}
\newcommand {\calZ} {\mathcal{Z}}
\newcommand {\calB} {\mathcal{B}}
\newcommand {\calC} {\mathcal{C}}
\newcommand {\calQ} {\mathcal{Q}}
\newcommand {\calV} {\mathcal{V}}
\newcommand {\calU} {\mathcal{U}}
\newcommand {\calS} {\mathcal{S}}

\DeclareMathOperator*{\argmax}{arg\,max}

\newcommand{\p}[1]{\left(#1\right)}
\newcommand{\pp}[1]{\left[#1\right]}
\newcommand{\ppp}[1]{\left\{#1\right\}}

\newcommand{\DIDeg}{\Delta I}

\newcommand{\DIDegStar}{\Delta I^{\ast}}

\newcommand {\bfalpha}{\bv{\alpha}}

\newcommand {\bfbeta}{\bv{\beta}}
\newcommand {\bfgamma}{\bv{\gamma}}
\newcommand {\bfzeta}{\bv{\zeta}}
\newcommand {\bfrho}{\bv{\rho}}

\newcommand{\reals}{\mathbb{R}}
\newcommand{\realsX}{\reals^{|\calX|}}
\newcommand{\realsXminusone}{\reals^{|\calX|-1}}
\newcommand{\realsK}{\mathbb{K}}
\newcommand{\realsKX}{\realsK^{|\calX|}}
\newcommand{\realsKplus}{\realsK_+}
\newcommand{\realsKplusX}{\realsKplus^{|\calX|}}

\newcommand{\calYSmall}{\calY_{\mathrm{small}}}
\newcommand{\calBK}{\calB_\realsK}

\newcommand{\rcritical}{r_\mathrm{critical}}

\newcommand{\xmax}{x_{\max}}

\newcommand{\omegaup}{\overline\omega}
\newcommand{\omegadown}{\underline\omega}

\newcommand{\density}{\varphi}

\newcommand{\weightQTag}{M\pp{\calQ'(\bfalpha,r)}}
\newcommand{\weightVTag}{M\pp{\calV'}}

